\documentclass[10pt,conference]{IEEEtran}
\usepackage{blindtext, graphicx}
\usepackage{cuted}

\newcommand{\erfc}{\operatorname{erfc}}

\usepackage{bigints}
\usepackage{caption}
\newlength\figureheight 
\newlength\figurewidth 
\usepackage{cite}
\usepackage{url}
\usepackage{blindtext, graphicx}
\usepackage{amsmath,amssymb,mathtools,bm,etoolbox}
\newcommand{\e}[1]{{\mathbb E}\left[ #1 \right]}
\usepackage{stackengine}
\def\delequal{\mathrel{\ensurestackMath{\stackon[1pt]{=}{\scriptstyle\Delta}}}}
\usepackage{suffix}
\DeclarePairedDelimiterX\MeijerM[3]{\lparen}{\rparen}%
{\begin{smallmatrix}#1 \\ #2\end{smallmatrix}\delimsize\vert\,#3}

\newcommand\MeijerG[8][]{%
  G^{\,#2,#3}_{#4,#5}\MeijerM[#1]{#6}{#7}{#8}}

\WithSuffix\newcommand\MeijerG*[7]{%
  G^{\,#1,#2}_{#3,#4}\MeijerM*{#5}{#6}{#7}}

\newcommand\ddfrac[2]{\frac{\displaystyle #1}{\displaystyle #2}}

\begin{document}
%
\title{Performance Analysis of Mixed RF/FSO Relaying under HPA Nonlinearity and IQ Imbalance}

\author{Elyes Balti, \textit{Student Member,} IEEE}
\maketitle
\begin{abstract}
In this paper, we present the performance analysis of asymmetric dual-hop RF/FSO system with multiple relays. The RF channels follow the correlated Rayleigh fading while the optical links are subject to the Gamma-Gamma fading. To select the candidate relay to forward the communication, we assume Partial Relay Selection (PRS) with outdated Channel State Information (CSI). Unlike the vast majority of work in this area, we introduce the impairments to the relays and the destination. We will propose three impairment models called Soft Envelope Limiter (SEL), Traveling Wave Tube Amplifier (TWTA) and IQ Imbalance in order to compare the resilience of our system with the RF one against the hardware impairments. Closed-from of the outage probability (OP) is derived in terms of Meijer's G function as well as the upper bound of the ergodic capacity (EC). The Bit Error Rate (BER) and the exact EC are evaluated numerically. Finally, analytical and numerical results are presented and validated by Monte Carlo simulation.
\end{abstract}


\IEEEpeerreviewmaketitle
\section{Introduction}
Wireless optical communications also known as Free-Space Optical (FSO) is considered as the key stone for the next generation of wireless communication since it has recently gained enormous attention for the vast majority of the most well-known networking applications such as fiber backup, disaster recoveries and redundant links \cite{1}. The main advantages of employing the FSO is to reduce the power consumption and provide higher bandwidth. Moreover, FSO becomes as an alternative or a complementary to the RF communication as it overcomes the problems of the spectrum scarcity and its licence access to free frequency band. In this context, many previous attempts have leveraged some these advantages by introducing the FSO into classical systems to be called Mixed RF/FSO systems. This new system architecture reduces not only the interferences level but also it offers full duplex Gegabit Ethernet throughput and high network security \cite{4}. Although the literature has shown the superiority of the mixed RF/FSO systems over the classical RF systems, they still suffer from the reliability scarcity and power efficient coverage. To overtake this difficulty, previous research attempts have proposed cooperative relaying techniques hybridized with the mixed RF/FSO systems since it improves not only the capacity of the wireless system but also it offers high Quality of Service (QoS). Recently, this new efficient system model has attracted considerable attention in particular using various relaying schemes. The most common used relaying techniques are Decode-and-Forward, Amplify-and-Forward, and Quantize-and-Encode \cite{9}, \cite{10}. Regarding the system with multiple relays, activating all relays to simultaneously forward the communication is not recommended because the problem of synchronization at the reception always occurs with optical communications. To solve this problem, only one relay is allowed to transmit the signal. In this case, a relay selection protocol is required to select this candidate relay. In the literature, there are many protocols previously proposed such as opportunistic relay selection, distributed switch and stay, max-select protocol \cite{12}. Unlike these protocols which require the knowledge of the total CSIs of the channels, Krikidis \textit{et al.} have proposed PRS in \cite{14} which requires the CSI of only one channel (source-relay or relay-destination). Unlike the slow time-varying channels, the rapid time-varying channels are characterized by high time-varying CSIs. In this case, the CSI used for relay selection is different from the CSI used for signal transmission, so the CSI is outdated due to the slow feedback coming from the relays. Unlike \cite{14} where the PRS is assumed with perfect CSI estimation, outdated CSI of Rayleigh and Nakagami-m fading is assumed in \cite{2} and \cite{16}. In spite of these considerable contributions in the area of mixed RF/FSO systems, they assumed ideal system without hardware impairments. In practice, however, the hardware (source, relays, destination) are susceptible to impairments, e.g., HPA non-linearities \cite{17} phase noise \cite{19} and IQ imbalance \cite{20}. Due to its low quality and price, the relay suffers from the non-linear PA impairment which is caused primarily by the non-linear amplification of the signal that may cause a distortion and a phase rotation of the signal. The most well-known non-linear PA model are TWTA, SEL \cite{22} and Ideal Soft Limiter Amplifier (ISLA) \cite{23}. Maletic \textit{et al.} \cite{22} concluded that the SEL has less severe impact on the system performance than the TWTA model. Furthermore, there are few attempts considering mixed RF/FSO system affected by a general model of impairments but they did not specify the type/nature of the hardware impairments. In this work, we propose a mixed RF/FSO system with multiple relays employing Fixed Gain (FG) relaying. PRS based on the CSI of the first hop is assumed with outdated CSI for relay selection and both the relays and the destination are respectively affected by non-linear power amplification (NLPA) and IQ imbalance. The rest of this paper is organized as follows: the system model is presented in section II. The performance analysis are detailed in section III. Section IV discusses the numerical and simulation results while the concluding remarks and the future directions are given in section V.
\section{System Model}
Our system consists of source ($S$), destination ($D$) and $N$ parallel relays wirelessly connected to the $S$ and $D$ shown by Fig.~1. For a given transmission, the source $S$ receives periodically the CSIs ($\gamma_{1(l)}$ for \textit{l} = 1\ldots \textit{N}) of the first hop from the \textit{N} relays and sorts them in an increasing order of magnitude as follows: $\gamma_{1(1)}\leq\gamma_{1(2)}\leq \ldots \leq\gamma_{1(N)}$. The perfect scenario is to select the best relay \textit{(m = N)} but this best one is not always available. In this case, $S$ will select the next best available relay. Consequently, PRS protocol selects the \textit{m}th worst or \textit{(N - m)}th best relay $R_{(m)}$. Given that the feedback sent from the relays to $S$ is susceptible to the delay, the CSI at the time of selection is different from the CSI at the instant of transmission. In this case, outdated CSI should be assumed instead of perfect CSI estimation. Hence, the instantaneous CSI used for relay selection $\tilde{\gamma}_{1(m)}$ and the instantaneous CSI $\gamma_{1(m)}$ used for transmission are correlated with the time correlation coefficient $\rho$.
\begin{center}
\includegraphics[width=8.5cm,height=6cm]{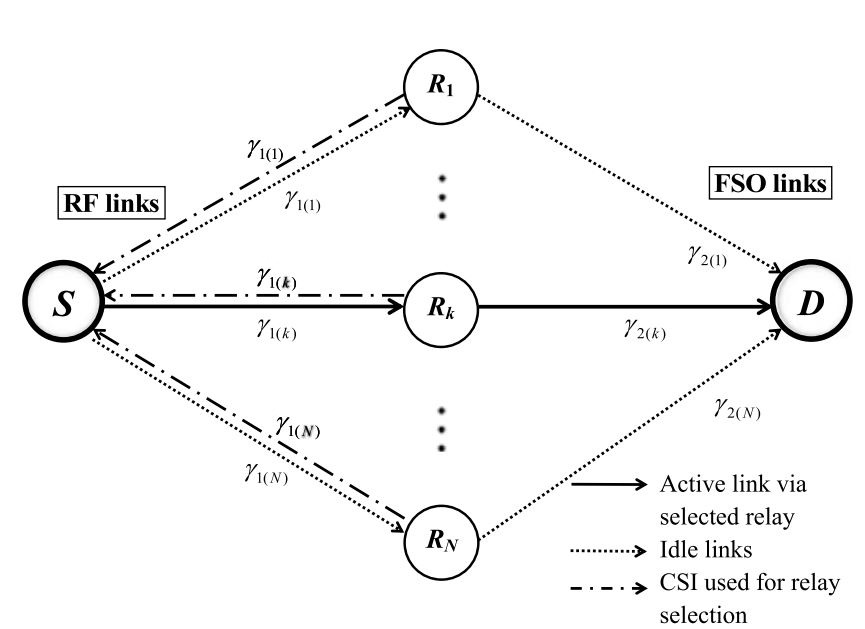}
\captionof{figure}{Mixed RF/FSO system with partial relay selection}
\label{fig1}
\end{center}
The received signal at the $m$th relay is given by
\begin{equation}
y_{1(m)} = h_m s + \nu_1,
\end{equation}
where $s \in \mathbb{C}$ is the information signal, $h_m$ is the RF fading between $S$ and $R_{(m)}$ and $\nu_1$  $\backsim$ $\mathcal{CN}$ (0, $\sigma^2_0$) is the AWGN of the RF channel.
\subsection{Relay's Power Amplifier non-linearity}
PA non-linearity impairment is introduced to the relays. The amplification of the signal happens in two time slots. In the first slot, the received signal at the relay $R_{(m)}$ is amplified by a proper gain $G$ as ${\phi_m} = G y_{1(m)}$. The gain $G$ can be defined as
\begin{equation}
G = \sqrt{\frac{\sigma^2}{\e{|h_m|^2} P_1 + \sigma^2_0}},
\end{equation}
where $\e{\cdot}$ is the expectation operator, $P_1$ is the average transmitted power from $S$ and $\sigma^2$ is the mean power of the signal at the output of the relay block. In the second time slot, the signal passes through a non-linear circuit $\psi_m = f(\phi_m)$.\\
The PA (Power Amplifier) of the relay is assumed to be memoryless. A memoryless PA is characterized by both Amplitude to Amplitude (AM/AM) and Amplitude to Phase (AM/PM) characteristics. The functions AM/AM and AM/PM transform the signal distortion respectively to $A_m(|\phi_m|)$ and $A_p(|\phi_m|)$ and then the output signal of the non-linear PA circuit is given by
\begin{equation}
\psi_m = A_m(|\phi_m|)~e^{j(\text{arg}(\phi_m)+A_p(|\phi_m|))},  
\end{equation}
where $\text{arg}(\phi_m)$ is the polar angle of the complex signal $\phi_m$. The characteristic functions of the SEL and TWTA impairments models are respectively given by \cite{17}
$$
A_m(|\phi_m|) =
\begin{cases}
|\phi_m| & \text{if}~~|\phi_m|~<~A_{sat} \\   
A_{sat} & \text{otherwise}
\end{cases}
\\
~,~~~~A_p(|\phi_m|) = 0,
$$
\begin{equation}
A_m(|\phi_m|) = \frac{A^2_{sat}|\phi_m|}{A^2_{sat} + |\phi_m|^2}~~,~~~~A_p(|\phi_m|)=\frac{\Phi_0~|\phi_m|^2}{A^2_{sat} + |\phi_m|^2},
\end{equation}
$A_{sat}$ is called the input saturation magnitude and $\Phi_0$ controls the maximum phase rotation. From a given saturation level $A_{sat}$, the relay's power amplifier operates at an input back-off (IBO), which is defined by IBO = $\frac{A^2_{sat}}{\sigma^2}$.\\
According to Bussgang Linearization theory \cite{17}, the output of the non-linear PA circuit linearly depends on both the linear scale $\delta$ of the input signal and a non-linear distortion $d$ which is uncorrelated with the input signal and follows the circularly complex Gaussian random variable $d \backsim \mathcal{CN} (0,~\sigma^2_d)$. Then, the AM/AM characteristic $A_m(|\phi_m|)$ can be expressed as follows
\begin{equation}
A_m(|\phi_m|) = \delta~x + d,    
\end{equation}
Regarding the SEL NLPA model, $\delta$ and $\sigma^2_d$ can be written as follows
\begin{equation}
\begin{split}
    &\delta = 1 - \exp\left(-\frac{A^2_{sat}}{\sigma^2}\right) + \frac{\sqrt{\pi} A_{sat}}{2\sigma^2}~\erfc\left(\frac{A_{sat}}{\sigma}\right),\\&
    \sigma^2_d = \sigma^2~\left[1 - \exp\left(-\frac{A^2_{sat}}{\sigma^2}\right) - \delta^2 \right],
\end{split}    
\end{equation}
The clipping factor $\xi$ of the SEL model is given by
\begin{equation}
\xi = 1 - \exp\left(-\frac{A^2_{sat}}{\sigma^2}\right),
\end{equation}
For the TWTA model, if the AM/PM effect of the characteristic $A_p(|\phi_m|)$ is neglected (i.e., $\Phi_0 \approx 0$), $\delta$ and $\sigma^2_d$ can be written as follows
\begin{equation}
\begin{split}
&\delta = \frac{A^2_{sat}}{\sigma^2}\left[1 + \frac{A^2_{sat}}{\sigma^2}\exp\left(\frac{A^2_{sat}}{\sigma^2}\right)  + \text{Ei}\left(-\frac{A^2_{sat}}{\sigma^2} \right)   \right],\\&
\sigma^2_d = -\frac{A^4_{sat}}{\sigma^2}\left[\left(1 + \frac{A^2_{sat}}{\sigma^2}\right)e^{\frac{A^2_{sat}}{\sigma^2}}\text{Ei}\left(-\frac{A^2_{sat}}{\sigma^2} \right) + 1 \right]-\sigma^2\delta^2,
\end{split}
\end{equation} 
where $\text{Ei}(\cdot)$ is the exponential integral function.\\
The clipping factor of the TWTA is given by
\begin{equation}
\xi = -\frac{A^4_{sat}}{\sigma^4}\left[\left(1+\frac{A^2_{sat}}{\sigma^2}\right)\exp\left(\frac{A^2_{sat}}{\sigma^2}\right)\text{Ei}\left(-\frac{A^2_{sat}}{\sigma^2} \right) + 1 \right],
\end{equation}
Then at the relay $R_{(m)}$, the RF amplified signal is converted to an optical one which is given by \cite{2}
\begin{equation}
r_{m} = G(1 + \eta \psi_m),
\end{equation}
where $\eta$ is the electrical-to-optical conversion coefficient.
\subsection{In-Phase and quadrature-phase imbalance at the destination}
In case of perfect IQ mismatch, the received signal at the destination can be expressed as follows
\begin{equation}
y_{2(m)} = I_m G \eta \psi_m + \nu_2,
\end{equation}
where $I_m$ is the optical irradiance between the relay $R_{(m)}$ and the destination $D$, $\eta$ is the optical-to-electrical conversion coefficient, and$ \nu_2$  $\backsim$ $\mathcal{CN}$ (0, $\sigma^2_0$) is the AWGN of the optical channels.\\
Given that the destination is affected by IQ imbalance, the received signal is given by
\begin{equation}
\hat{y}_{2(m)} = \omega_1 y_{2(m)} + \omega_2 (y_{2(m)})^*,
\end{equation}
where $(y_{2(m)})^*$ is called the mirror signal introduced by the IQ imbalance at $D$ and the coefficients $\omega_1$ and $\omega_2$ are respectively given by
\begin{equation}
\begin{split}
\omega_1 = \frac{1 + \zeta e^{-j\theta}}{2},~~
\omega_2 = \frac{1 - \zeta e^{j\theta}}{2},
\end{split}
\end{equation}
where $\theta$ and $\zeta$ are respectively the phase and the magnitude imbalance. This impairment is modeled by the Image-Leakage Ratio (ILR), which is given by $\text{ILR} = \left|\frac{\omega_1}{\omega_2}\right|^2$.\\
For an ideal $D$, $\theta = 0, \zeta = 1, \omega_1 = 1, \omega_2 = 0$, and ILR = 0.
\subsection{Channels Models}
Since the RF channels are subject to correlated Rayleigh fading, the probability density function (PDF) and the cumulative distribution function (CDF) of the instantaneous RF SNR $\gamma_{1(m)}$ are respectively given by \cite{2}
\begin{equation}
\begin{split}
f_{\gamma_{1(m)}}(x) = m{N \choose m}\sum_{n=0}^{m-1} \frac{(-1)^n}{[(N-m+n)(1-\rho)+1]\overline{\gamma}_1}\\ 
\times~{m-1 \choose n} \exp\left(-\frac{ (N-m+n+1)x }{((N-m+n)(1-\rho)+1)\overline{\gamma}_1}\right),
\end{split}
\end{equation}
\begin{equation}
\begin{split}
F_{\gamma_{1(m)}}(x) = 1 - m{N \choose m}\sum_{n=0}^{m-1} \frac{(-1)^n}{N-m+n+1}~~~~~~~~~~\\~ 
\times{m-1 \choose n} \exp\left(-\frac{ (N-m+n+1)x }{((N-m+n)(1-\rho)+1)\overline{\gamma}_1}\right),
\end{split}
\end{equation}
Since the instantaneous SNR $\gamma_{2(m)}$ experiences Gamma-Gamma fading, its PDF is given by
\begin{equation}
f_{\gamma_{2(m)}}(x) = \frac{(\alpha\beta)^{\frac{\alpha+\beta}{2}} x^{\frac{\alpha+\beta }{4}-1}}{\Gamma(\alpha)\Gamma(\beta)\overline{\gamma}_2^{\frac{\alpha+\beta}{4}}} K_{\alpha-\beta}\left(2\sqrt{\alpha\beta\sqrt{\frac{x}{\overline{\gamma}_2}}}\right),
\end{equation}
where $K_{\nu}(\cdot)$ is the $\nu$-th order modified Bessel function of the second kind, $\alpha$ and $\beta$ are respectively the small-scale and large-scale of the scattering process in the atmospheric environment. These parameters are given by
\begin{equation}
\begin{split}
\alpha = \left( \exp\left[ \frac{0.49 \sigma_R^2}{(1+1.11\sigma_R^{\frac{12}{5}})^{\frac{7}{6}}}\right] -1\right)^{-1}, \\
\beta = \left( \exp\left[ \frac{0.51 \sigma_R^2}{(1+0.69\sigma_R^{\frac{12}{5}})^{\frac{5}{6}}}\right] -1\right)^{-1},
\end{split}
\end{equation}
where $\sigma_R^2$ is called Rytov variance which is a metric of the atmospheric turbulence intensity.
\subsection{End-to-end signal-to-noise-plus-distortion ratio (SNDR)}
The average SNR of the first hop is given by
\begin{equation}
\overline{\gamma}_1 = \frac{P_1|h_m|^2}{\sigma_0^2},    
\end{equation}
While the average SNR $\overline{\gamma}_2$\footnote[1]{The average SNR $\overline{\gamma}_2$ is defined as $\overline{\gamma}_2 = \eta^2\e{I_{m}^2}/\sigma_{0}^2$, while the average electrical SNR $\mu_2$ is given by $\mu_2 = \eta^2\e{I_{m}}^2/\sigma_{0}^2$. Therefore, the relation between the average SNR and the average electrical SNR is trivial given that $ \frac{\e{I^2_{m}}}{\e{I_{m}}^2} = \sigma^2_{\text{si}} + 1$, where $\sigma^2_{\text{si}}$ is the scintillation index \cite{scin}.} of the second hop can be expressed as
\begin{align}
 \overline{\gamma}_2 = \frac{\e{I^2_{m}}}{\e{I_{m}}^2}\mu_2,   
\end{align}
where $\mu_2$ is the average electrical SNR given by
\begin{equation}
    \mu_2 = \frac{\eta^2\e{I_{m}}^2}{\sigma_{0}^2},
\end{equation}
According to \cite[Eq.~(16)]{22}, the end-to-end SNDR is given by Eq.~(21).
\section{Performance Analysis}
In this section, we present the analysis of the OP, the BER and the EC. We will show that the OP and BER are limited by irreducible floors and the EC is finite and saturated by a ceiling at the high SNRs values. The floors and the ceiling are certainly caused by the hardware impairments originating from the relays and the destination.\\
\subsection{Outage Probability Analysis}
The outage probability is defined as the probability that the end-to-end SNDR falls below a given outage threshold $\gamma_{\text{th}}$. It can be written as follows
\begin{equation}
P_{\text{out}}(\gamma_{\text{th}}) \delequal \text{Pr}[\gamma_{\text{e2e}} < \gamma_{\text{th}}],
\end{equation}
where \text{Pr($\cdot$)} is the probability notation. The analytical expression of the SNDR given by Eq.~(21) should be placed in Eq.~(18). After some algebraic manipulations, the OP can be expressed by Eq.~(22). Note that the CDF $F_{\gamma_{1(m)}}$ is defined only if $1 - \text{ILR}\gamma_{\text{th}} > 0$, otherwise it is equal to a unity.
The term $\kappa$ is the ratio between the received SNR and the average transmitted SNDR at the relay which is given by
\begin{equation}
\kappa = 1 + \frac{\sigma_d^2}{\delta^2G^2\sigma_0^2},
\end{equation}
Note that the OP is equal to Eq.~(24) for $\gamma_{\text{th}} < \frac{1}{\text{ILR}}$, otherwise, it is equal to a unity.\\
Since the expression of the outage probability involves complex function such as the Meijer-G function, we need to derive an asymptotic high SNR expression to unpack engineering insights about the system gain. Given that the outage performance saturates at high SNR by the outage floor caused by the hardware impairments, it is trivial to conclude that the diversity gain $G_d$ is equal to zero. For an ideal hardware and after expanding the Meijer-G function at high SNR using \cite[Eq.~(07.34.06.0001.01)]{24}, it can be shown that the diversity gain is given by
\begin{equation}
G_d = \min\left(1,~\frac{\alpha}{2},~\frac{\beta}{2}\right),  
\end{equation}
\subsection{Average Bit Error Rate}
The BER can be expressed as follows
\begin{equation}
\overline{P_e} = \frac{q^p}{2\Gamma(p)}\int\limits_{0}^{\infty}\gamma^{p-1}e^{-q\gamma}~F_{\gamma_{\text{e2e}}}(\gamma)~d\gamma,  
\end{equation}
where $F_{\gamma_{\text{e2e}}}(\cdot)$ is the CDF of $\gamma_{\text{e2e}}$, $p$ and $q$ are the parameters that indicate the modulation format, respectively. As we mentioned earlier, the mathematical terms related to the impairments render the integral calculus very complex. As a result, deriving a closed-form of the BER is not possible. In this case, a numerical integration is required. Note that a floor occurs at high SNRs values which prevents the BER from converging to zero. This floor will be shown graphically later in the section of numerical results.
\newpage
\begin{strip}
\noindent\rule{18.5cm}{1pt}
\vspace*{0.5cm}
\begin{equation}
\gamma_{e2e} = \frac{\gamma_{1(m)}\gamma_{2(m)}}{\text{ILR}\gamma_{1(m)}\gamma_{2(m)}+(1+\text{ILR})\kappa\gamma_{2(m)}+(1+\text{ILR})(\e{\gamma_{1(m)}} + \kappa)},
\end{equation}
\noindent\rule{18.5cm}{1pt}
\begin{equation}
\begin{split}
P_{\text{out}}(\gamma_{\text{th}}) =& 1 - \frac{2^{\alpha+\beta-2}}{\pi\Gamma(\alpha)\Gamma(\beta)}m{N \choose m} \sum_{n=0}^{m-1} {m-1 \choose n}\frac{(-1)^n}{N-m+n+1}
\exp\left(-\frac{(N-m+n+1)\kappa(1+\text{ILR})\gamma_{\text{th}}}{((N-m+n)(1-\rho)+1)(1-\text{ILR}\gamma_{\text{th}})\overline{\gamma}_1}\right)\\&
\times~\MeijerG[\Bigg]{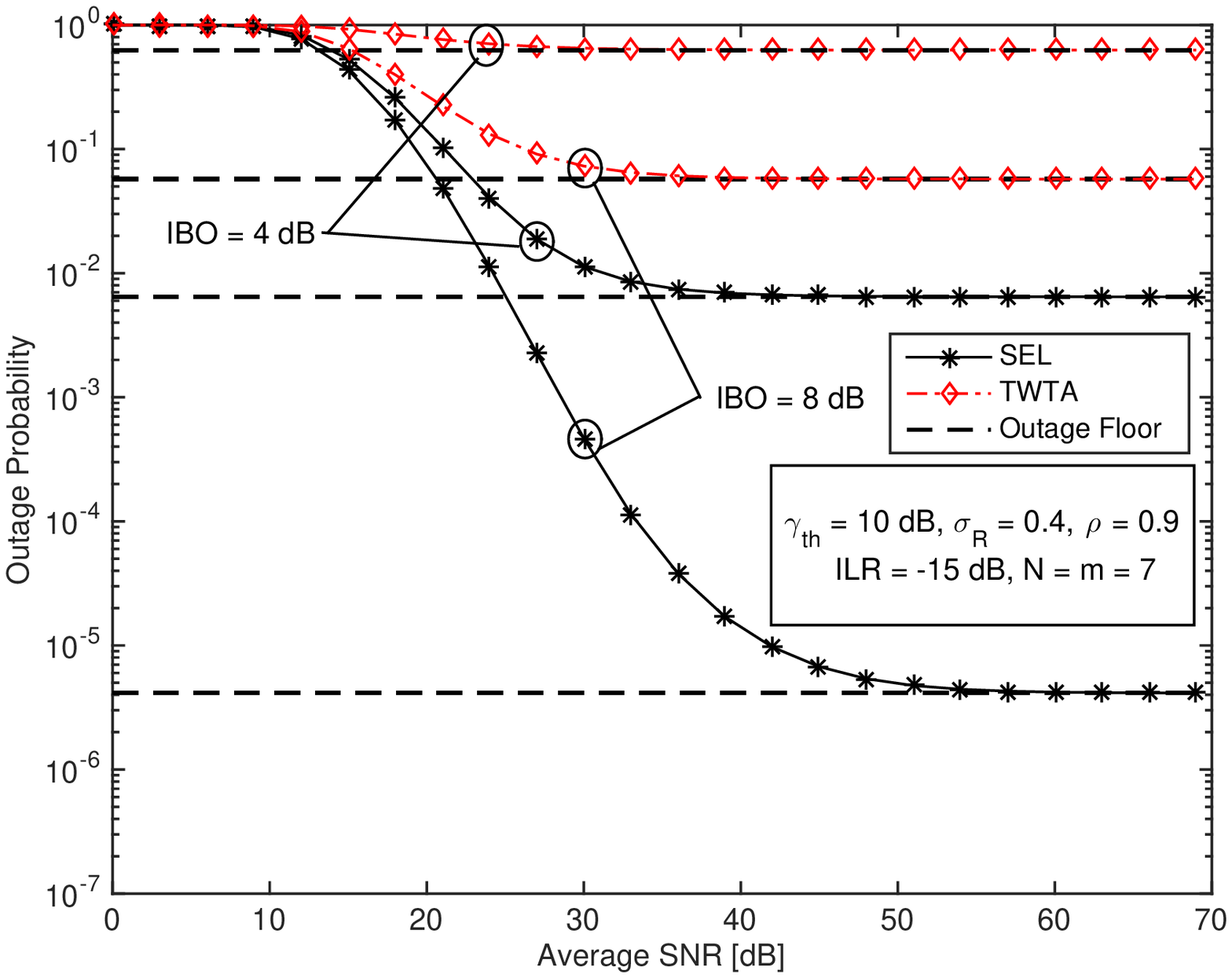}{0}{0}{5}{-}{\frac{\alpha}{2}, \frac{\alpha+1}{2}, \frac{\beta}{2}, \frac{\beta+1}{2}, 0}{\frac{(\alpha\beta)^2(\e{\gamma_{1(m)}}+\kappa)(N-m+n+1)\gamma_{\text{th}}}{16((N-m+n)(1-\rho)+1)(1-\text{ILR}\gamma_{\text{th}})\overline{\gamma}_1 \overline{\gamma}_2}},
\end{split}
\end{equation}
\noindent\rule{18.5cm}{1pt}
\end{strip}
\vspace*{-1cm}
\subsection{Ergodic Capacity}
The ergodic capacity, expressed in bit/s/Hz, is defined as the  maximum error-free data transferred by the channel of the system. It can be written as follows
\begin{equation}
\overline{C} \delequal \frac{1}{2}\e{\log_{2}(1+\gamma_{\text{e2e}})},
\end{equation}
The capacity can be calculated by deriving the PDF of the SNDR. However, an exact closed-form of Eq.~(26) is very difficult due to the mathematical terms related to the impairments. To evaluate the system capacity, we should refer to the numerical integration.\\
In spite of the difficulty to calculate an exact closed-form of the EC, we can derive a simpler expression by referring to the approximation given by \cite[Eq.~(27)]{22}
\begin{equation}
\e{\log_2\left(1 + \frac{\psi}{\varphi}\right)} \approx \log_2\left(1+\frac{\e{\psi}}{\e{\varphi}} \right),
\end{equation}
For high SNR values, the SNDR converges to $\gamma^*$ defined by
\begin{equation}
\lim_{\overline{\gamma}_1,\overline{\gamma}_2\to\infty}\gamma_{\text{e2e}} = \frac{1}{\frac{(1+\text{ILR})\xi}{\delta}-1} = \gamma^*,
\end{equation}
\newtheorem{thm}{Theorem}
\newtheorem{cor}[thm]{Corollary}
\begin{cor}
Suppose that $\overline{\gamma}_1$ and $\overline{\gamma}_2$ converge to infinity and the electrical and optical channels are independent, the ergodic capacity converges to a capacity ceiling defined by
\begin{equation}
\overline{C}^* = \frac{1}{2}\log_2(1+\gamma^*), 
\end{equation}
\end{cor}
\begin{proof}
Since the SNDR converges to $\gamma^*$ as the average SNRs of the first and second hop largely increase, the dominated convergence theorem allows to move the limit inside the logarithm function.
\end{proof}
If the relaying system is linear, i.e, the system is only impaired by IQ imbalance, the SNDR and the average capacity are saturated at the high SNR regime as follows
\begin{equation}
\gamma^* = \frac{1}{\text{ILR}},~~\overline{C}^{*} = \frac{1}{2}\log_2\left(1+\frac{1}{\text{ILR}}\right),
\end{equation}
To characterize the EC, it is possible to derive the expression of the upper bound stated by the following theorem.
\newtheorem{theorem}{Theorem}
\begin{theorem}
For asymmetric (Rayleigh/Gamma-Gamma) fading channels, the ergodic capacity $\overline{C}$ for non-ideal hardware is upper bounded by
\begin{equation}
    \overline{C} \leq \frac{1}{2}\log_{2}\left(1 + \frac{\mathcal{J}}{\text{ILR}~\mathcal{J} + 1}\right)
\end{equation}
\end{theorem}
where $\mathcal{J}$ is given by
\begin{equation}
\begin{split}
 \mathcal{J} = \e{\frac{\gamma_{1(m)}\gamma_{2(m)}}{\text{ILR}\gamma_{2(m)} + \tau }},
 \end{split}
\end{equation}
where $\tau = (1+\text{ILR})\kappa\gamma_{2(m)}+(1+\text{ILR})(\e{\gamma_{1(m)}} + \kappa)$.\\
After some mathematical manipulations, $\mathcal{J}$ is given by
\begin{equation}
\begin{split}
&\mathcal{J} = \ddfrac{m{N \choose m}(\alpha\beta)^\frac{\alpha + \beta}{2}\left( \frac{\e{\gamma_{1(m)}} + \kappa}{\kappa}\right)^\frac{\alpha+\beta}{4}}{2\pi(1+\text{ILR})\kappa\Gamma(\alpha)\Gamma(\beta)\overline{\gamma}_2^\frac{\alpha+\beta}{4}}\\&~
\times \sum_{n=0}^{m-1} {m-1 \choose n}\frac{(-1)^m((N-m+n)(1-\rho)+1)\overline{\gamma}_1}{(N-m+n+1)^2}\\&
\times~\MeijerG[\Bigg]{5}{1}{1}{5}{\kappa_0}{\kappa_1}
{\frac{(\alpha\beta)^2(\e{\gamma_{1(m)}}+\kappa)}{16\kappa\overline{\gamma}_2}},
\end{split}
\end{equation}
where $\kappa_0, \kappa_1$ are given by
\begin{equation}
\begin{split}
&\kappa_0 = -\frac{\alpha + \beta}{4},\\&
\kappa_1 = \left[\frac{\alpha - \beta}{4}, \frac{\alpha - \beta + 2}{4}, \frac{\beta - \alpha}{4}, \frac{\beta - \alpha +2 }{4},-\frac{\alpha + \beta}{4}\right],
\end{split}
\end{equation}
\section{Numerical results}
This section presents analytical and numerical \footnote{For all cases, $10^9$ realizations of the random variables were generated to perform the Monte Carlo simulation in MATLAB.} results of the OP, BER and EC obtained from the mathematical expressions mentioned in the previous section.\\
Since the RF channel experiences correlated Rayleigh fading, it can be generated using the algorithm in \cite{26}. The atmospheric turbulence follows Gamma-Gamma fading, which can be generated by using the formula, $I = I_{X}I_{Y}$, where $I_{X}$ and $I_{Y}$ are independent random variables, which follow Gamma distribution.
\begin{table}[!ht]
\renewcommand{\arraystretch}{1.3}
\caption{Simulation Parameters}
\label{tab:example}
\centering
\begin{tabular}{|c|c|}
    \hline 
    \textbf{Parameters}  &  \textbf{Values} \\
    \hline \hline
    Outage threshold $\gamma_{\text{th}}$ [dB]   &   10\\
    Time correlation $\rho$    &  0.9  \\
    Number of relays $N$    &  7  \\
    Rank of selected relay $k$    &  7  \\
    Rytov variance $\sigma^2_{R}$ & 0.16 \\
    Modulation & BPSK \\
    \hline 
\end{tabular}
\end{table}
\vspace*{-0.5cm}
\begin{center}
\includegraphics[width=8.5cm,height=6cm]{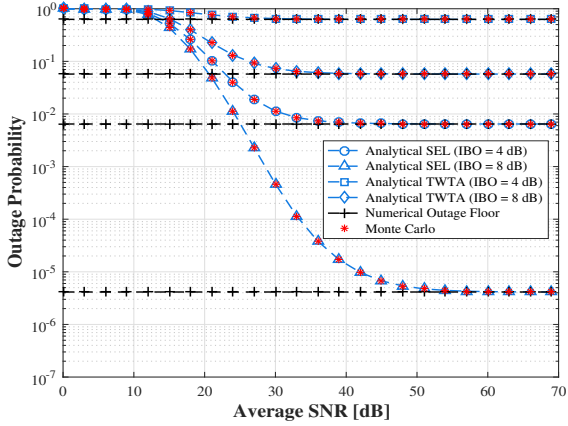}
\captionof{figure}{Outage probability versus the average SNR for different values of IBO}
\label{fig1}
\end{center}
The dependence of the OP with respect to the average SNR for the case of SEL and TWTA NLPA models are shown in Fig.~2. For ILR = -15 dB, this value can be obtained by 1 dB of amplitude imbalance and 15$^{\circ}$ of phase imbalance. We observe that when the relays's system are impaired by SEL, the OP is lower compared to the case of the TWTA impairment. For example, for SNR = 60 dB and IBO = 8 dB, the OPs under the effect of SEL and TWTA are respectively equal to 4~10$^{-6}$ and 6~10$^{-2}$. Moreover, we observe that the system performs better as long as the IBO value increases. As the average SNR per hop increases, the outage floors appear for both cases SEL and TWTA but the system performs better under the effect of SEL than TWTA. Therefore the TWTA has more severe impact on the system performance than the SEL. In addition, our system apprears to be more resilient to the hardware imperfections compared to the system assumed in \cite{22}. In fact, for (IBO, ILR, $\gamma_{\text{th}}$, SNR) equal to (8, -15, 10, 45)[dB] and under the joint effect of SEL and IQ imbalance, the OP of our system is equal to 7~10$^{-6}$. However, the OP of the system suggested by Maletic \textit{et al.} in \cite{22} is equal to 6~10$^{-3}$ shown by Fig.~4. The factors that achieve this significant enhancement of our system over the classical RF system \cite{22} are essentially the FSO technology and the diversity of the RF part characterized by the multiple relays.
\begin{center}
\includegraphics[width=8.5cm,height=6cm]{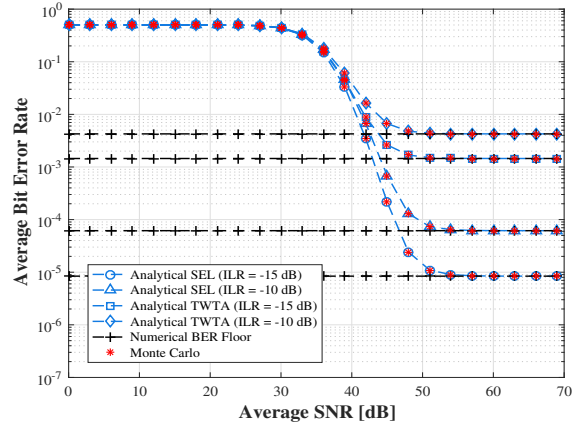}
\captionof{figure}{Average Bit Error Rate versus the average SNR for different values of ILR}
\label{fig1}
\end{center}
\vspace*{-0.5cm}
\begin{center}
\includegraphics[width=8.5cm,height=6cm]{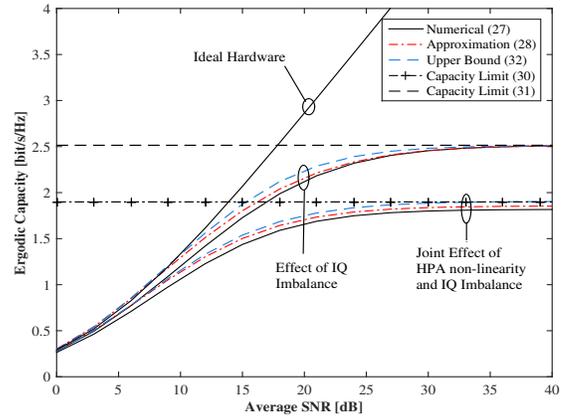}
\captionof{figure}{Exact, approximate and upper bound of the ergodic capacity versus the average SNR}
\label{fig1}
\end{center}
Fig.~3 shows the variations of the BER against the average SNR per hop for different values of the ILR. We clearly observe that the BER is limited by an irreducible floor caused by the joint effect of HPA non-linearity and IQ imbalance. As the ILR value increases, the destination is more susceptible to impairments and hence the BER performance deteriorates further. As a comparison with the work done by \cite{22}, for ILR = -15 dB, IBO = 5 dB, BPSK modulation and assuming the TWTA impairment, the mixed RF/FSO system outperforms the classical RF suggested by \cite{22}. In fact, for an average SNR equal to 45 dB and ILR = -15 dB, the BER performance of our system is equal to 7~10$^{-3}$, however, the BER of the RF system is approximately equal to 1.9~10$^{-2}$ shown by Fig.~9 \cite{22}. Regarding the impact of the SEL impairments, our system again performs better and proves its high resiliency against the imperfections than the RF system. In fact, for the same previous configuration of ILR, IBO and modulation format, the BER is equal to 2~10$^{-4}$ while the BER for RF system is equal to 1.3~10$^{-3}$. Note that even our system is impaired by TWTA, the most severe impairments, there is no significant difference between the BER of mixed RF/FSO and the BER of the full RF system under the effect of the SEL, the less severe impairments. We conclude that our mixed RF/FSO system is more robust to the impairments than the previous RF system due to the advantages brought by the FSO technique.\\
The variations of the EC versus the average SNR hop assuming linear and non-linear relaying (SEL model) with an impaired destination are shown in Fig.~4. As expected, the system operating with linear relaying outperforms the system performance in the case of non-linear relaying. As
the impairments at the relays disappear, the saturation level of the capacity increases but the capacity is still limited by a ceiling superior than the ceiling of the capacity under the joint effect of the NLPA relaying and IQI. The significant difference between the two ECs at high SNR shows clearly the deleterious effect of the high power amplifier non-linearities on the system performance.\\
Note that the capacity ceiling is independent on the system parameters and it depends only on the impairments parameters (ILR, IBO), that is why the ceiling level is still the same for the ceiling suggested by \cite{22}. The main advantage of our system compared with the RF system, is that even though the two systems are limited by the same ceiling level, the mixed RF/FSO system capacity increases faster than the capacity of the RF system.
\section{Conclusion}
In this work, we provided the analysis of various models of impairments and their effects on the system performance. We introduced the SEL and TWTA as HPA non-linearities affecting the relays and we assume that $D$ is impaired by IQ imbalance. We studied the effects of these hardware imperfections on the system performance in terms of OP, BER and EC. We concluded that the system performs better as the IBO increases and the ILR decreases. Moreover, it turned out that the TWTA has more severe impact on the system performance than the SEL model. Furthermore, even though the performance deteriorates under the effects of the imperfections, we noted that the introduction of the FSO technique makes the mixed RF/FSO system more resilient to the hardware impairments than the previous RF relaying system. As future directions, unlike the previous work that developed various techniques for the impairments compensation, we intend to develop an algorithm/technique that must remove completely or at least to a large extent the residual impairments which still cause the performance deterioration.
\bibliographystyle{IEEEtran}
\bibliography{bibliography}
\end{document}